\documentclass[sigconf, anonymous=false]{acmart}

\usepackage[utf8]{inputenc} % allow utf-8 input
\usepackage[T1]{fontenc}    % use 8-bit T1 fonts
\usepackage{hyperref}       % hyperlinks
\usepackage{url}            % simple URL typesetting
\usepackage{booktabs}       % professional-quality tables
\usepackage{amsfonts}       % blackboard math symbols
\usepackage{nicefrac}       % compact symbols for 1/2, etc.
\usepackage{microtype}      % microtypography
\usepackage{bm}
\usepackage{subfig}
\usepackage{optidef}
\usepackage{amsfonts}
\usepackage{amsmath}
\usepackage{amsthm}
\usepackage{multirow}
\usepackage{multicol}
\usepackage{algorithm, algpseudocode}
\newtheorem{theorem}[]{Theorem}
\newtheorem{lemma}[]{Lemma}

\newcommand{\notimplies}{\;\not\!\!\!\implies}

\newcommand{\rb}[1]{\rotatebox{0}{ #1 }}

\renewcommand{\vec}[1]{\bm {#1}} 
\newcommand{\ev}[2]{\mathbb E_{#1} \left [ #2 \right ]}

\AtBeginDocument{%
  \providecommand\BibTeX{{%
    \normalfont B\kern-0.5em{\scshape i\kern-0.25em b}\kern-0.8em\TeX}}}

\setcopyright{acmcopyright}
\copyrightyear{2020}
\acmYear{2020}
\acmDOI{10.1145/1122445.1122456}

\acmConference[Preprint]{}{May 2021}
\begin{document}
\title{Stochastic Opinion Dynamics for Interest Prediction in Social Networks}

\author{Marios Papachristou}
\affiliation{
  \institution{Cornell University}
}
\email{papachristoumarios@cs.cornell.edu} 
\authornote{Work done when M.P. was an undergraduate student at NTUA.}

\author{Dimitris Fotakis}
\affiliation{
  \institution{National Technical University of Athens}
}
\email{fotakis@cs.ntua.gr}

\renewcommand{\shortauthors}{Papachristou, and Fotakis}

\begin{abstract}

We exploit the core-periphery structure and the strong homophilic properties of online social networks to develop faster and more accurate algorithms for user interest prediction. The core of modern social networks consists of relatively few influential users, whose interest profiles are publicly available, while the majority of peripheral users follow enough of them based on common interests. Our approach is to predict the interests of the peripheral nodes starting from the interests of their influential connections. To this end, we need a formal model that explains how common interests lead to network connections. Thus, we propose a stochastic interest formation model, the Nearest Neighbor Influence Model (NNIM), which is inspired by the Hegselmann-Krause opinion formation model and aims to explain how homophily shapes the network. Based on NNIM, we develop an efficient approach for predicting the interests of the peripheral users. At the technical level, we use Variational Expectation-Maximization to optimize the instantaneous likelihood function using a mean-field approximation of NNIM. We prove that our algorithm converges fast and is capable of scaling smoothly to networks with millions of nodes. Our experiments on standard network benchmarks demonstrate that our algorithm runs up to two orders of magnitude faster than the best known node embedding methods and achieves similar accuracy. 
\end{abstract}




\maketitle

\begin{flushright}
\small {
\emph{ ``Birds of a feather flock together.'' 
--- Plato, Symposium, ca. 385 BC}
}
\end{flushright}

\section{Introduction}
\label{s:intro}

Most modern large-scale Online Social Networks (OSN) exhibit the so-called \emph{core-periphery structure} (see e.g., \cite{coreperiphery_benson, coreperiphery_stochasticblockmodel, pareto1935mind, coreperiphery_stochasticblockmodel, snyder1979structural, nemeth1985international, wallerstein1987world, papachristou2021sublinear} and the references therein). Namely, their nodes are naturally  partitioned into a \emph{core set} $C$ of nodes tightly connected with each other, and a \emph{periphery set} $U$, where the nodes are sparsely connected, but are relatively well-connected to the core. In most cases (see also our analysis in Fig.~\ref{fig:coverage}), the core nodes almost dominate the rest of the network, in the sense that a small fraction of $\delta n$ high-degree nodes dominate an $(1 - a) n$ fraction of the network's engaged nodes (where ``engaged'' refers to nodes with degree above than a small threshold). If we restrict to engaged nodes only, even a sublinear fraction of nodes dominate almost everything (see also \cite{mgeop, mgeop2, mgeop3}). These influential core nodes, which posses a large number of incoming connections, or \emph{followers}, are also known (and serve) as the \emph{celebrities} or the \emph{influencers} of the network. Influencers tend to publicly expose --- mainly for commercial reasons \cite{khamis2017self, freberg2011social} --- their profile information (friends and interests). %Thus, information about can be gathered easily.

Another major driving force shaping the structure of social networks is \emph{homophily}, i.e., the property under which connected individuals in a social network have similar interests \cite{blau1, blau2}. Modern large-scale OSN seem to exhibit strong homophilic trends, which was a significant part of our motivation. 

\smallskip\noindent{\bf Approach and Contribution.}
In this work, we leverage homophilic trends and the core-periphery structure of modern OSN to obtain scalable and accurate methods for predicting the interests of the peripheral users. Our approach is to identify and use the influencers of the network as \emph{trend-initializers}. Then, we let the subnetwork consisting of the periphery nodes evolve according to an iterative process initialized to an aggregation of the influencers' features. %The influencers' sublinear number allows for a quite fast initialization (in worst-case strongly subquadratic-time) of the users' interests. 
%Inspired by \emph{coevolutionary opinion formation} \cite{hk,BGM13}, we next treat the network as the result of a natural \emph{interest exchange} dynamical process. %where each peripheral user updates her features according to the labels of her $k$-nearest neighbors in the periphery, until %\emph{consensus} (convergence) 
%the process converges. 

At the conceptual level, we consider a network consisting of some \emph{core users} and the users that follow them, which correspond to the \emph{peripheral users}. In real-world scenaria, we can imagine having a very large social network, and the core users being famous athletes, politicians, singers, fashion models, etc. These users have exposed interests (or labels), which we want to use in order to infer the labels of the peripheral users, for which we do not have information (e.g., due to privacy policies). So, we gather the follower relations between the peripheral users and the core users, and construct a \emph{bipartite graph} with one set being the core users and the other set being the peripheral users. The core set is very small -- in reality \emph{sublinear} -- and the induced bipartite subgraph contains a \emph{very small} fraction of the edges that typically decreases, as the network size increases. This very small fraction of edges on the induced bipartite subgraph can consistently be observed in all of our experiments (see the last row on Table \ref{tab:results}). Moreover, another key reason for which we choose to work with this limited information is that the network may be huge, thus it may be prohibiting to mine all -- or at least a large fraction -- of its edges. In fact, this is a serious practical consideration, taking into account the massive growth of modern OSN, and a considerable volume of work \cite{benson2019link, liben2007link, rhodes2015inferring} has devoted to addressing it. The influencers' sublinear number allows for a quite fast initialization (in worst-case strongly subquadratic-time) of the peripheral users' interests. 
The information provided by aggregating the interests of the core users followed by each peripheral user is a reasonable starting point. However, it does not account for peripheral \emph{user-to-user interactions}, which can be observed in the original network. To explain the missing links between peripheral users and to account for their influence in the peripheral users' interests, we develop a \emph{homophily-based stochastic opinion formation model}, which we call the \emph{Nearest Neighbor Influence Model} (NNIM). NNIM's state at \emph{consensus} (or close to it) accounts for the final interests of the peripheral users, while their neighborhoods act as a replacement for the actual ones. %Consensus state is reached via a \emph{stopping criterion}. 
%We make the model and the stopping criterion rigorous below.

More precisely, the Nearest Neighbor Influence Model is a stochastic iterative process. according to which the peripheral users evolve their binary interest vectors. At each timestep, each peripheral user samples a new binary interest vector based on the interests of her $k$ nearest neighbors (wrt. their interest vectors) in the periphery. The general structure of NNIM is inspired by the Hegselmann-Krausse model \cite{hk}. However, NNIM is stochastic and is used as a \emph{generative model}, aiming to explain, through homophily, the coevolution of the network structure and the peripheral user interests (see Table~\ref{tab:stats} and the last paragraph of Section~\ref{s:model}). 

From a more technical viewpoint, our prediction method aims to recover the latent NNIM interest vectors of the peripheral users that maximize the likelihood that NNIM evolves through the periphery nodes as observed. Although the idea is simple, its efficient implementation requires significant effort and care (see Section~\ref{s:analysis}). We use Variational Expectation-Maximization, due to the latent nature of NNIM, since direct maximization of the log-likelihood is intractable. As a result, we obtain a simplified mean-field approximation of NNIM (see the 3rd column of Algorithm~1,  and \eqref{eq:trhk}), which is similar to the classical opinion dynamics equations, thus establishing a connection between stochastic and deterministic opinion dynamics. We prove (Theorem~\ref{thm:convergence}) that our algorithm converges in a finite number of steps and establish an upper bound between the total variation distance, the number of iterations, and the number $k$ of neighbors used in the interest exchange processes.% (which affects the running time). %Our algorithm efficiently scales to networks with millions of nodes. 

Our user interest prediction method scales smoothly to networks with millions of nodes, with an \emph{almost linear-time complexity with respect to the network size} and \emph{linear-time with respect to the feature dimension $d$}, for appropriate choices of hyperparameters. Key to our algorithm's scalability is that throughout the NNIM process, each peripheral user interacts only with her $k$-nearest neighbors. We evaluated our method experimentally on six standard network benchmarks taken from \cite{snap,dblp1, dblp2} with quite different characteristics (see Table~\ref{tab:stats}). Our experimental results suggest that our method performs similarly (or often outperforms) sophisticated node embedding and traditional opinion dynamics methods in terms of AUC-ROC and RMSE, whilst being able to run up to $100$ times faster than the best known node embedding methods in networks with up to $10^6$ nodes (see Table~\ref{tab:results}). Compared against more standard baselines, like Collaborative Filtering and Label Propagation, which enjoy similar running time to our algorithm, our approach achieves significantly better accuracy in the most interesting datasets. 

Our work draws ideas from (and contributes to) three major research directions (see also Section~\ref{s:related}). From an algorithmic perspective, we take advantage of the core-periphery structure of OSN to speed up inference in large-scale networks. Moreover, we introduce and analyze a natural stochastic generalization of coevolutionary opinion dynamics, which we eventually utilize for user interest prediction. As a result, we obtain a new truly scalable user prediction approach with excellent accuracy. Our methodology can be extended to a variety of problems in combinatorial optimization and machine learning, where inference from the entire network leads to prohibitive running times\footnote{\textbf{Our implementation is anonymously released in \cite{source_code}.}}.

\subsection{Related Work}
\label{s:related}

\textbf{Core-periphery structure} of networks has mainly gathered attention from socio-economical \cite{wallerstein1987world, krugman1991increasing, maccormack2010architecture} and network modeling perspectives \cite{nemeth1985international, coreperiphery_stochasticblockmodel, benson2019link}. Computer science literature is mostly concentrated in learning core-periphery models. From an algorithmic perspective, the closest work to ours is \cite{avin2017distributed}, where Avin et al. show how to speed up tasks in a distributed setting. However, they do not provide an algorithm for efficiently identifying the core in large networks, as we do in this work. Finally, the work of \cite{papachristou2021sublinear} is close to ours as it proposes a generative model for core-periphery networks that exhibits a core of sublinear size (wrt to the size of the network) that acts as an almost dominating set of the network, fits the model to real-world small-scale networks, and compares with \cite{coreperiphery_benson, tudisco2019nonlinear}. 

\textbf{Opinion Dynamics} models have been around for decades. The best known are the DeGroot model \cite{de2009morality}, the Friedkin-Johnsen (FJ) model \cite{fj}, and the HK model \cite{hk}. Our NNIM model is conceptually close to those in \cite{hk,BGM13,chazelle2011total,chazelle2016inertial}, where the agent opinions evolve as a discrete dynamical system and the opinions at the next step result from an aggregation of ones' and her neighbors' opinions, where the neighborhood is built dynamically from the observations at the current step. NNIM is also similar to the Random HK model \cite{local_interactions}, where each agent chooses uniformly at random $k$ neighbors from a ball of radius $\varepsilon$ centered at her opinion. In NNIM, each peripheral user chooses her $k$ nearest neighbors instead (e.g., as in the $k$-NN model \cite{BGM13}). NNIM provides a \emph{stochastic variant} of known \emph{coevolutionary opinion formation} models, thus generalizing existing deterministic ones. 
At the conceptual level, our work follows a significant and prolific research direction investigating the dynamics and the social and algorithmic efficiency of natural learning and opinion formation processes in social networks \cite{AcemogluO11,ChazelleW19,Danescu-Niculescu-MizilKKL09,GaitondeKT20,MolaviERJ16}. Building on previous work, we show how to efficiently exploit the results of such social learning processes for user interest prediction in OSN. Similarly to our work, \cite{DeVGBG16} aims to predict user opinions over time from a history of noisy signals of their neighbors' opinions. \cite{DeVGBG16} focuses on how opinions evolve over time, identifying conditions and aims at models that scale well wrt event sequence (not the size of the network). In our work, however, we use the consensus state for user interest prediction and aim at models that scale well wrt the size of the network.

\textbf{Multilabel classification} in graphs has a relatively long history. To begin with, the classical work on label propagation \cite{raghavan2007near} infers community memberships in networks via propagating labels between the nodes until a consensus is reached. Besides, similar work in \cite{egonets, bigclam} devises a random graph model to classify nodes with features within communities. Moreover, the upsurge of embedding methods, which use random walks,  matrix factorization-based learning objectives, or signal processing transformations \cite{node2vec, deepwalk, graphwave, nodesketch, graphsage} has been used for multilabel classification. Multilabel classification with embeddings as a \emph{standardized benchmark task} for evaluating embedding methods uses them as inputs to a supervised model, usually \emph{logistic regression}. The input graph nodes typically have features in a high-dimensional space, whereas the target labels lie in a low-dimensional space. In contrast, in our work, inputs and outputs have the same dimensionality. 
 
Our work is also related to \textbf{inference in probabilistic graphical models} with latent variables and with a likelihood that cannot be computed in a computationally efficient manner, because integration for the latent variables significantly affects the running time. Some characteristic examples are the MAG Model in OSN \cite{mag, mag2} and training of HMMs \cite{baum} with the EM algorithm \cite{em}. The EM algorithm maximizes the expectation of the joint likelihood of the data by imposing a distribution over the latent variables. We use the mean-field approximation in our paper \cite{mean_field_1, mean_field_2}, a technique that is widely used in the statistical physics community.

\section{The Nearest Neighbor Influence Model}
\label{s:model}

As discussed in Section~\ref{s:intro}, we assume that the network $G(C \cup U, E, \hat {\vec X})$ consists of a core $C$, a periphery $U$ with size $|U| = n$, and a matrix of initial features $\hat {\vec X}$ with an $d$-dimensional binary vector $\hat {\vec X}_c$ for each $c \in C$ which represents the trends that $c$ endorses throughout the iterative process. The core members serve as \emph{trend-initializers} meaning that their labels do not change throughout the process. NNIM proceeds in steps, where we use the letter $t$ to denote time steps. Each peripheral user $u \in U$ has a $d$-dimensional vector at time $t$, denoted by $\vec X_u^{(t)}$. Each $u \in U$ initializes her vector as a Bernoulli trial with a probability equal to the maximum likelihood estimation (sample mean) given the members of the core she follows. At each step $t \ge 1$ each member of the periphery $u \in U$ observes her $k$-nearest neighbors in the periphery $U$ with respect to the Hamming Norm $\sum_{i = 1}^d \mathbf 1 \left \{ X_{ui}^{(t)} \neq X_{vi}^{(t)} \right \}$, which quantifies how much the agent disagrees with another agent $v \in U$, and constructs her stochastic set $\mathcal K^{(t)}(u)$ . Afterwards the agent constructs the vector $\vec \xi_u^{(t + 1)}$ which is the average of the observed opinions inside the set $\mathcal K^{(t)}(u)$ including the user herself, as $\vec \xi_u^{(t + 1)} = \frac 1 k \sum_{v \in \mathcal K^{(t)}(u)} \vec X_v^{(t)}$. Then each agent updates her opinion $\vec X_u^{(t + 1)}$ at time $t + 1$ drawing a Bernoulli sample from $\vec {Be}(\vec \xi_u^{(t + 1)})$, independently for each coordinate. Consensus is defined by a \emph{stopping criterion} for the stochastic model for an accuracy parameter $\epsilon$, and is given as $\tau(\epsilon) = \inf \{ t \ge 0 | \| \mathbf E [\vec X^{(t + 1)}] - \mathbf E [\vec X^{(t)}] \|_1 \le \epsilon \}$. This criterion agrees with the classical statistical distance criterion since $\| \mathbf E [\vec X^{(t + 1)}] - \mathbf E [\vec X^{(t)}] \|_{1, 1} = \sum_{i = 1}^d \sum_{u \in U} d_{TV} (X_{iu}^{(t + 1)}, X_{iu}^{(t)})$ where $d_{TV}(\cdot, \cdot)$ denotes the total variation distance. This way, at each step $t$, a \emph{stochastic temporal graph} $G_t$ is created, where each agent has a neighborhood that corresponds to her $k$-nearest neighbors, in place of the actual OSN (see also the \textsc{\texttt{NNIM}} procedure in Algorithm~1). A very simple \emph{example} could be the following, suppose that we have 3 agents $v_1, v_2, v_3$ with initial probabilities $(1/2, 1/2, 1/2)^T$ and $k = 2$. The realizations of the opinions are $(0, 1, 1)^T$. The probabilities of the next round are $(0 + 1) / 2 = 1/2$ for $v_1$, and $(1 + 1) / 2 = 1$ for $v_2$ and $v_3$. The new round starts with the agents flipping coins with biases $(1/2, 1, 1)^T$. Suppose that the new realizations are $(1, 1, 1)^T$ and thus the new parameters are $(1, 1, 1)^T$, hence $\tau(1/2) = 2$. 

Intuitively, NNIM aims to explain the space of user labels in the network by homophily. So, NNIM treats the $k$ nearest neighbors of a user wrt. her labels as her highly homophilic nodes. To test this hypothesis, we compare the actual neighborhood of the ground social network with the $k$-nearest-neighbors-neighborhood for each $u \in C \cup U$ . Given the un-initialized directed social network $G(C \cup U, E, \vec {\hat X})$ (where each user has a binary interest vector), we define $\vec \alpha_w = \frac 1 {|N^+(w)| + 1} \sum_{v \in N^+(w) \cup \{ w \} } \hat {\vec X}_v$ and $\vec \beta_w = \frac 1 {k_w} \sum_{v \in \mathcal K(w)} \hat {\vec X}_v$, where $N^+(w)$ is the set of users that $w$ follows and $k_w$ is the number of nearest neighbors we choose for each node $w$ (in our experiment $k_w$ takes the value $|N^+(w)|+1$ or $\lceil \log n \rceil$).  These vectors measure the aggregate effect of the neighborhood (actual or due to $k_w$-nearest-neighbors) on determining the interest vector of $w$ and, therefore, similar values of these two vectors correspond to similar effects on the interest vector of $w$. For this reason, we measure the Root Mean Squared Error $\mathrm {RMSE}(\vec \alpha_w, \vec \beta_w) = d^{-1/2} \| \vec \alpha_w - \vec \beta_w \|$ for each node $w \in C \cup U$. Then, we take a degree weighted average, where the weight of each node is $(1 + |N^+(w)|) / (|E| + |C \cup U|)$, and measure the distance from 100\%. This degree-weighted average puts emphasis on the nodes by order of ``prestige'' in the network $G$. We call this quantity the \emph{Homophilic Index} (HI) of $G$. Intuitively, the HI measures how much the aforementioned two neighborhoods look similar in the feature space (see Table \ref{tab:stats}).

\begin{table}
\centering
\caption{Dataset Statistics and Homophilic Index for values of nearest neighbors $k_u$ being the outdegree (including the user) ($k_{u1} = |N^+(u)| + 1$) and $k_{u2} = \lceil \log n \rceil$.}
\small
\begin{tabular}{llrrrrr}
\toprule
Name & Type & Nodes & Edges & $d$ & $k_{u1}$ & $k_{u2}$ \\
\midrule
facebook \cite{snap, egonets} & ego  & 1.03K & 27.8K  & 576 & 93.24 & 91.03 \\
dblp-dyn \cite{dblp2} & co-authors & 1.23K & 4.6K & 43  & 82.02 & 83.56 \\
fb-pages \cite{snap, musae} & page-page  & 22.5K & 342K & 4 & 91.69 & 92.31 \\
github \cite{snap, musae} & developer  & 37.7K & 578K & 1 & 85.48 & 84.41 \\
dblp \cite{dblp1} & co-authors  & 41.3K & 420K & 29 & 82.54 & 85.62 \\
pokec \cite{snap, pokec} & social  & 1.6M & 30.6M & 280 & 66.10 & 67.72\\
\bottomrule    
\end{tabular}

\label{tab:stats}
\end{table}

\subsection{Model Inference through Variational Expectation-Maximization}
\label{s:analysis}

For the inference problem we are interested in determining the parameters the peripheral nodes in the NNIM model, namely the probability vectors $\{ \vec \xi_u^{(t)} \}_{u \in U}^{t \ge 0}$ of the feature vectors $\{ \vec X_u^{(t)} \}_{u \in U}^{t \ge 0}$ given the initial state of the cores' labels.  We start by forming the optimization objective (log-likelihood) at each step $t$. Initially, according to our setting we assume that we know the initial values of the peripheral user labels as the samples with probabilities equal to the sample average of the influencers of the core she is following, as delineated in  the procedure \textsc{\texttt{Initialize\_Infer}} of Algorithm 1. In reference \cite{bindel2015bad}, Bindel et al. view the opinion formation problem for the FJ model under a \emph{game-theoretical viewpoint} where each agent suffers a quadratic convex cost for not reaching consensus at a given time $t$. Similarly, in our case at each time $t$ is the (instantaneous) log-likelihood that better explains the distribution of the agents parametrized by $\vec \xi^{(t + 1)}$ is needed to be maximized, given the previous state of the agents $\vec X^{(t)}$, that is 

\begin{equation}
    \mathcal L_{\xi}^{(t + 1)} \left (\vec  \xi^{(t + 1)} \right ) = \log \sum_{\vec X^{(t)}} \Pr \left [\vec X^{(t)} | \vec \xi^{(t + 1)} \right ] 
\end{equation}

We observe the initial opinions $\vec X^{(0)}$ of the network and then the opinion vectors are latent, thus inference requires summation over exponentially many events. The opinion vectors are assumed to have the \emph{Markov property}, namely the opinions at a given time are affected only by the previous step. Observe that the stochastic nature of the model imposes intractability on the likelihood functions $\mathcal L_{\xi}^{(t + 1)}$ since it requires a summation over the exponentially-many latent variables $\vec X^{(t)}$ which have binary outcomes. For simplicity, we assume that the interest distribution is approximated by a \emph{variational distribution} $Q^{(t)}$ that makes the latent variables $\{ \vec X^{(t)} \}_{t \ge 1}$ independent, and approaches the actual parameters $\{ \vec \xi^{(t)} \}_{t \ge 1}$ having a form of $Q^{(t)} =  \prod_{u \in U} \prod_{i = 1}^d \left (\phi_{iu}^{(t)} \right )^{X_{iu}^{(t)}} \left (1 - \phi_{iu}^{(t)} \right )^{1 - X_{iu}^{(t)}}$ where $\vec \phi_u^{(t)}$ are the variational parameters that are the ``empirical counterparts'' of the actual parameters $\vec \xi_u^{(t)}$. Using Jensen's Inequality on the likelihood function $\mathcal L_{\xi}^{(t + 1)}$, that is 

\begin{equation}
    \mathcal L_\xi^{(t + 1)} \ge  \mathbb E_{Q^{(t)}} \bigg [\log \Pr [\vec X^{(t)} \big | \xi^{(t + 1)}] \bigg ] + \ev {Q^{(t)}} {-\log Q(\vec X^{(t)})}
\end{equation}

we obtain two terms, the first of which -- referred as the variational lower bound (VB) -- we maximize\footnote{The second term (entropy) is positive and depends only on step $t$.}. The maximization of the VB $\mathcal L_{Q, \xi}^{(t + 1)} = \ev {Q^{(t)}} {\log \Pr  [\vec X^{(t)} \big | \vec \xi^{(t + 1)} ]}$ is a tractable problem \cite{hoffman2013stochastic, mag2} and can be used as a proxy for approximating the actual interest distribution. It can be expressed as

\begin{equation*} \label{eq:large_bound}
\begin{split}
    \mathcal L_{Q, \xi}^{(t + 1)}  = & \mathbb E_{Q^{(t)}} \bigg [ \sum_{i = 1}^d \sum_{u \in U} \sum_{v \in U} \mathbf 1 \left \{ v \in \mathcal K^{(t)}(u) \right \} \\  & \bigg ( X_{iv}^{(t)} \log \xi_{iu}^{(t + 1)}  + \, \left (1 - X_{iv}^{(t)} \right ) \log \left (1 - \xi_{iu}^{(t + 1)}  \right ) \bigg ) \bigg ]
\end{split}
\end{equation*}

Computing the expectation over the stochastic set $\mathcal K^{(t)}(u)$ of the $k$-nearest neighbors exactly still poses computational barriers. We approximate \eqref{eq:large_bound} by choosing the $k$-nearest neighbors in the parameter space.\footnote{It is expected that $\mathcal K^{(t)}(u)$ and  $K^{(t)}(u)$ have significant overlap (see Appendix).} Subsequently, the VB can be expressed as

\begin{equation} \label{approx_lambda}
\begin{split}
    \mathcal L_{Q, \xi}^{(t + 1)} \approx & \sum_{i = 1}^d \sum_{u \in U} \sum_{v \in K^{(t)}(u)} \bigg [ \phi_{iv}^{(t)} \log \phi_{iu}^{(t + 1)} + \\ & \left (1 - \phi_{iv}^{(t)} \right ) \log \left (1 - \phi_{iu}^{(t + 1)} \right ) \bigg ]
\end{split}
\end{equation}

By setting the gradient to zero, we get the following \emph{Inference Algorithm}: 

\begin{enumerate}
    \item Identify the core $C$ of the network (see Sec.~\ref{sec:experimental_setting}). 
    \item \emph{(Optional)} Perform PCA on the feature vectors of the core nodes. 
    \item We initialize each agent's initial value $\vec \phi_{u}^{(0)}$ with the average of the (reduced) feature vectors of the influencers she follows.
    \item Until convergence we perform the following update rule, that comes from setting the gradient of Eq. \ref{approx_lambda} to zero:
    \begin{equation} \label{eq:trhk}
        \vec \phi_u^{(t + 1)} = \frac {1} {k} \sum_{v \in K^{(t)}(u)} \vec \phi_v^{(t)}
    \end{equation}
    \item \emph{(Optional)} Project back to the original space using the inverse PCA transformation and clip values that fall outside $[0, 1]$.
\end{enumerate}

This system of equations rise by observing the instantaneous likelihood at each time $t$.
We can perform \emph{regularization} to the model by adding extra opinions --- for instance, the initial state --- with weights to the model. 
Ending, we define the ``macroscopic'' distribution which is parametrized by $\{ \vec \mu^{(t)} \}_{t \ge 1}$ and has a Bernoulli density over the labels, with parameter vectors defined as $\vec \mu^{(t)} = \frac 1 n \sum_{u \in U} \vec \xi_u^{(t)}$ and displays how the agents behave with respect to trends in general, namely if they adopt (or not) an interest as a whole. Given the calculated parameters $\vec \phi^{(t + 1)}$, we can determine the parameters $\vec \mu^{(t + 1)}$ using the same variational approach. More specifically, the expected log-likelihood $\mathcal L_{Q, \mu}^{(t)}$ of the \emph{macroscopic parameters} $\vec \mu^{(t)}$ under the variational distribution $Q$ is given as $\mathcal L_{Q, \mu}^{(t)} = \sum_{u \in U} \sum_{i = 1}^d \left ( \phi_{iu}^{(t)}  \log \mu_{i}^{(t)} + \left (1 - \phi_{iu}^{(t)} \right ) \log \left ( 1 - \mu_{i}^{(t)} \right ) \right ) $. Invoking the expected value according to the variational parameters and setting $\partial \mathcal L_{Q, \mu}^{(t)}  / \partial \mu_{i}^{(t)} = 0$ for all $1 \le i \le d$ and $1 \le t \le T$. Analogously to \eqref{eq:trhk}, we obtain the update rule  $\vec \mu^{(t)} =  \sum_{u \in U} \vec \phi_u^{(t)} / n$. 

\textbf{Relation to EM.} We refer to the above equations as the \emph{mean field equations} since the variational parameters are ``approximated'' with exactly the same model, but now the process does not involve randomness.  From an EM perspective, we can view our algorithm as having two discrete steps: In the E-step we compute the $k$ nearest neighbors of each agent whereas in the M-step we update the variational parameters by averaging and then compute the ``macroscopic distribution'' by averaging on the new variational parameters per dimension. The form of \eqref{eq:trhk} is very familiar to the classical opinion dynamics equations, like the HK model. We also prove the following theorem regarding finite-time convergence of \textsc{\texttt{Inference\_NNIM}} and about the convergence rate behaviour. 
The former part of the result is proven using Lyapunov Stability Theory and the latter part concerns the convergence rate of a Markov Chain with $k$-regular transition matrices of the sequence $\{ G_t \}_{t \ge 1}$ (see Appendix A). 

\begin{theorem}\label{thm:convergence}
    The system of \eqref{eq:trhk} converges in finite time under any consistent total ordering. Moreover, it suffices to perform $T = O \left ( \log (1 / D) / \log k \right )$ iterations to make the total variation distance between the current state and the consensus state is less than $d \cdot D$.  
\end{theorem}

\textbf{Per-step Cost.} We use Locality Sensitive Hashing (LSH) \cite{gionis1999similarity} with accuracy $\epsilon > 0$ to construct the nearest neighbor sets, that yields an  almost-linear in $n$ and linear in $d$ and $k$ per-step cost of 
% $O \left ( n^{1 + 1 / (1 + \epsilon)^2} d k \log(1 / D) \log^{-1} k  \right ) = \\ 
$\widetilde O \left  ( n^{1 + 1 / (1 + \epsilon)^2} d k  \right ) $ 
that scales in large real-world networks.

\textbf{Regularization.} In order to make our model more ``stubborn'' to the initial opinions of the agents we can impose regularization functions $\omega (t)$ such that the negative cross-entropy between the current opinions and the initial opinions is maximized, that is 
$\omega(t) = \alpha \sum_{u \in U} \sum_{i = 1}^d \big [ \phi_{iu}^{(0)} \log \phi_{iu}^{(t)} + \big (1 - \phi_{iu}^{(0)} \big ) \log \big ( 1 - \phi_{iu}^{(t)} \big ) \big ]
$ where $\alpha$ is the regularization parameter. Intuitively, we introduce one more sample to our model that is modeled by the initial conditions. Differentiating the likelihood we arrive at the recurrence relation $\vec \phi_u^{(t + 1)} = (k + \alpha)^{-1} \big [\sum_{v \in K^{(t)}(u)} \vec \phi_v^{(t)} +  \alpha \vec \phi_u^{(0)} \big ] $. This equation is similar to the opinion dynamics model where each agent is ``stubborn'' --- namely stuck to her initial opinion --- with a weight $\alpha$ as an input, such as in  \cite{fj}.

\section{Experiments}

\subsection{Datasets}

We perform experiments to test the \emph{ranking quality} of our system as well as its \emph{accuracy}, on networks of various sizes and from multiple disciplines.  Each of these datasets were obtained by \cite{snap, dblp1, dblp2} and are available online (see Table \ref{tab:stats}). The datasets are the following

\begin{itemize}

\item \emph{facebook \cite{snap, egonets}.} Contains an ego-network of user 107 in the Facebook network. Friendships in Facebook are undirected. To avoid the obvious domination by the ego node, we have removed the outgoing links of the ego node and kept the incoming links. We predict (anonymized) attributes of users.

\item \emph{dblp-dyn \cite{dblp2}.} Co-authorship graph from DBLP (author have at least 10 publications) for 43 publication venues in the period 1994-1998. We predict publication venues. 
% Vertices of the graph are authors and an edge exists between them if the corresponding authors have written a paper together in a given period of time. Only authors who had at least 10 publications (in a selected set of 43 conferences/journals) from 1990 to 2010 are considered. There are in total 2,723 authors. Each vertex at each time is associated to a set of 43 attributes corresponding to the number of publications in each conference/journal during the related period. We have chosen to keep the period between 1994 and 1998.

\item \emph{facebook-pages \cite{snap, musae}.} A page-page graph of verified Facebook pages, with nodes representing pages and links are mutual likes between them. The pages belong to four categories defined by facebook (politicians, governmental organizations, television shows, companies). 

\item \emph{github \cite{snap, musae}.} Social network of GitHub developers as of June 2019 who have starred at least 10 repositories and edges are mutual follower relations between them. All users in this dataset have one label, whether the user is a web or a machine learning developer.

\item \emph{dblp \cite{dblp1}.} This data set depicts a co-authorship graph built from the DBLP digital library between January 1990 and February 2011. The labels represent 29 publication venues (conferences, journals). 

\item \emph{pokec \cite{snap,pokec}.} Pokec is an anonymized social network with 1.6 million users. We extracted the hobbies from the users profiles and kept the 280 most common hobbies (so that everyone is covered by at least one label). We removed the nodes that have not disclosed their profile information and connections. The final sub-network has a size of 533K nodes\footnote{Order of magnitude is $10^6$ nodes}.

\end{itemize}

\subsection{Experimental Setting} \label{sec:experimental_setting}

For each of the methods we are compared with, we focus on a common same-input-same-output task. Given the binary labeled network $G$, and a target size for the \emph{core set}, we gather the core $C$ that contains the influencers of the network. We build the bipartite graph with $C$ and the set $U$ of the peripheral users, who follow those in $C$. We use the information from the core users to predict the labels of the peripheral users in $U$ (all labels have the same dimension).%, which constitutes the rest of the network. 
Then, we assign a $d$-dimensional vector of probabilities (scores) to each peripheral user $u \in U$, each coordinate of which corresponds to the probability that the given label is 1.  
% We also want to note that in a realistic setting we want to be able to find influential users in a network (possibly in an online manner) using limited information from the network, which itself may have up to \emph{millions} of users. 
% The bipartite network we build has a significantly smaller size than the actual network and is easier to be obtained than the actual network, for example if all the edges in the original network are impossible to be sampled and thus are missing (see e.g. \cite{benson2019link, rhodes2015inferring}) . 
This setting serves as a common starting point for the majority of the experiments.     

The problem of obtaining the core of influencers is similar to the \emph{Maximum Coverage} (MC) \cite{nemhauser1978analysis, feige1998threshold}, and the greedy algorithm which proceeds in rounds and chooses the node with the maximum number of uncovered neighbors yields a $(1 - 1/e)$-approximation. Running the greedy algorithm ad-hoc has (i) high computational cost in large networks; (ii) after some iterations it may favour nodes that are not ``core'' to the network -- but contribute to the greedy covering -- resulting in poorer features for the prediction task (i.e. peripheral users following very few core nodes). To establish a trade-off between having  a good coverage and good predictions efficiently, we reside on a fork of the original algorithm which we call Bucketed Greedy Bucketed MC (BGMC). In the BGMC setting, we have an upper bound $K$ of nodes we want to use in our coverage. We sort the nodes according to their in-degree and put them into $\log (n / K) / \log \gamma$ non-uniform buckets $V_1, \dots, V_r, \dots $ of sizes $\lceil \gamma K \rceil, \dots, \lceil \gamma^r  K \rceil - \lceil \gamma^{r-1} K \rceil, \dots$, for some $\gamma > 1$. We then start by constraining the neighborhoods of vertices to $V_1$ and run the greedy maximum coverage algorithm on it. If we either cover all the nodes or exhaust the $K$ choices, we return. Otherwise, we iterate using the set $V_2$, and so on, via removing the already covered nodes at each iteration. Although it is evident that the BGMC algorithm does not in general yield a solution set that equals the conventional greedy solution and has a strictly lesser approximation ratio, the algorithm yields remarkably good results when run on OSN. More specifically, for an outdegree threshold value $\tau = 4$ (users with outdegree less than $\tau = 4$ are omitted) a population of $n^{0.7}$ influencers dominate about $74.01 \pm 14.91 \%$ of the networks in question (see the \emph{Coverage} row in Table \ref{tab:results} and Figure \ref{fig:coverage}).\footnote{It is important to note that we experimented with various randomized policies e.g. from \cite{molnar}, however, the results were inferior, both in terms of coverage and in terms of ranking quality and accuracy.}
 
 \begin{figure}
    \centering
    \includegraphics[width=0.45\textwidth]{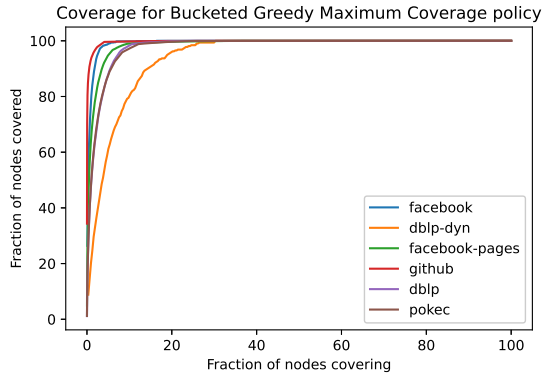}
    \caption{Coverage curve for the BGMC policy for $\tau = 4$.}
    \label{fig:coverage}
\end{figure}

After obtaining the bipartite graph which has one of its sides known to us, and the other is to be predicted, we fit the following algorithms to predict the peripheral user labels

\begin{itemize}
    \item \emph{Collaborative Filtering (via the core).} We use a simple \emph{colla\-borative-filtering-based} approach as a baseline: The labels are propagated from the core nodes to the peripheral nodes, and for each peripheral node the probability that a specific label is 1 is equal to the sample mean of the core nodes she is following. The reason for including this baseline -- which is equivalent to the initialization step of NNIM (before the \texttt{NNIM\_Inference} method)  -- is because we want to measure the positive contribution of the NNIM procedure -- which does a kind of ``sophisticated smoothing'' -- compared to a \emph{static} one. 
    
    \item \emph{Opinion Dynamics Models.} Each peripheral user gets the sample mean of the influencers she's following and then updates her vector using an update rule. We refer to two methods. The former one, NNIM (our proposed method) uses the $k$ Nearest Neighbors  of each user (including herself) and aggregates the result for the next iteration according to \eqref{eq:trhk}. The latter is the Random HK model of \cite{local_interactions} where each agent chooses a random subset of size $k$ from a set of neighbors that are contained in a certain radius $\varepsilon$. For NNIM we experiment with $k \in \{ \lceil \log n \rceil, \lceil \sqrt n \rceil \}$ with (and without) regularization. In the case of regularization we extra static opinions at each iteration with some weight parameter $\alpha$. In our case, we add the sample mean from the influencers, which the particular user is following, at each iteration. This regularization scheme bears a resemblance to the notion of \emph{close-minded agents} in \cite{chazelle2011total,chazelle2016inertial}.  
    We use LSH to (approximately) construct the $k$ nearest neighbor sets efficiently\footnote{We have experimented with exact algorithms, and had similar experimental results in terms of the measured metrics.}. Additionally, to avoid dealing with high dimensionality prior to running the \textsc{\texttt{Inference\_nnim}} procedure, we perform dimensionality reduction (PCA) keeping a 95\% of the explained variance. After running the algorithm, we invert the transformation and clip the variables that fall outside $[0, 1]$. For Random HK, we perform experiments with $k = \lceil \log n \rceil$ and $\varepsilon = \sqrt {d / 2}$. 
    \item \emph{Node Embeddings.} We train\footnote{We convert the directed graph to undirected} node2vec \cite{node2vec}, GraphWave \cite{graphwave} and NodeSketch \cite{nodesketch} embeddings on the same graph and then fit a multilabel logistic regression model. We chose node2vec as a classical random-walk-based approach, GraphWave as a transformation-based approach, and NodeSketch which is a new method based on recursive sketching.
\end{itemize}

For completeness, we also compare with the following off-the-shelf methods applied with knowledge of the \emph{whole network} in order to measure -- up to some extent -- the effect of the network  

\begin{itemize}
    \item \emph{Label Propagation.} We run the label propagation algorithm of \cite{raghavan2007near} where initially the core nodes have labels.
    
    \item \emph{Collaborative Filtering (Dynamic).} The algorithm is similar to \cite{raghavan2007near}: Instead of using the max frequency label at each iteration for each user, we take a mean of the labels of the properly labeled neighbors. Initially, only the core node have labels. This approach is also referred in \cite{de2009morality}. 

\end{itemize}

At this stage we have a ground truth matrix of size $|U| \times d$ with binary entries, each of which representing a whether or not each label is present in each peripheral user, and a $|U| \times d$ matrix with entries in $[0, 1]$ that correspond to the predicted probabilities that each peripheral user will have a certain label  present to her vector. We focus on the following metrics to quantify our findings 

\begin{itemize}
    \item \emph{AUC-ROC.} We use AUC-ROC to argue about the quality of our predictions as a ranking. he AUC-ROC measure is that is serves as a metric for a \emph{bipartite ranking}, that is the AUC-ROC is higher when positive labels (in the ground truth) are ranked above negative labels in terms of the attributes probabilities.\footnote{Alternatively, the AUC-ROC measure equals the area under the true positives-false positives curve for every probability threshold value $\theta \in [0, 1]$. }. We measure the micro-averaged AUC-ROC\footnote{In terms of macro-averaged AUC-ROC, our method is uniformly better as well.} between the two matrices (ground truth and predictions). We measure AUC-ROC in the following occasions: Between all the labels of the two matrices, and between the top 50\% prevalent labels of the ground truth. We believe that these measurements reflect multiple occasions of ranking that we want to perform since, for instance, when recommending items/labels to users, we want to have a good ranking for a percentile of the labels. Furthermore, we avoid measuring the AUC-ROC in ``degenerate'' settings, where, for example, having some labels equal to 0 for a large percentage of the users yields a very high AUC-ROC due to negative labels being placed correctly; which is of course not representative of the reality, since positive labels may be misplaced, even at the presence of a high enough AUC-ROC. 
    \item \emph{$F_1$-Score.} For the Label Propagation experiment, since all outcomes are binary, we report the micro $F_1$-Score.     
    \item \emph{RMSE.} We perform a row-wise mean operation on the ground truth and the predicted matrices and obtain two $d$-dimensional vectors that represent the \emph{macroscopic} distribution (i.e. the distribution vector across all users), and calculate the RMSE between these two vectors. Here, the RMSE metric is used to quantify the \emph{accuracy} of each experiment (closeness of predictions to ground truth). The row-wise mean is performed in order to transform the two quantities to the same domain, i.e. $[0, 1]^d$ instead of $\{ 0, 1 \}^{|U| \times d}$ and $[0, 1]^{|U| \times d}$.
    \item \emph{Running Time.} We measure the running time in seconds to demonstrate the scaling of the largest dataset (pokec). The running time of the smaller datasets is omitted, since there is no clear winner in terms of running time. 
    \item \emph{Coverage \& Influencers Core Size.} We report the core size as a percentage of the total size of the network and the number of covered nodes (by the core users) w.r.t. the size of the network.  
\end{itemize}

\section{Discussion, Impact \& Conclusion}

Regarding baselines, the baselines on the entire network (Collaborative Filtering (Dynamic) and Label Propagation) yield poor results compared to the other methods. More importantly, there is even large deviation compared to the Collaborative Filtering benchmark on the bipartite graph between the core and the peripheral users, perhaps, due to highly correlated interests, which result in poorer estimations of the scores. Furthermore, the interests from the core users only, are very good starting points for estimation of the interest distributions. One possible explanation of this accords with the idea of influential users in a network who have strong opinions and tend to have uncorrelated interests with respect to the other core users, such as in the case of two strongly opinionated presidential candidates. Hence they serve as good initialization points for the initialization of NNIM. The NNIM then serves as a simulator of the peripheral interactions of users based on highly homophilic interactions, and improves both AUC-ROC and RMSE, via performing a ``smoothing'' procedure. Also, in fb-pages, collaborative filtering from the core users produces very good results, comparable to our method, because the page categories (politicians, governmental organizations, television shows, companies) are initially  well-separated, therefore, the initial estimation step produces good estimates, since features are almost orthogonal. However, in the other datasets, where the components of the bipartite graph are not well separated and the dimensionality is higher, NNIM surpasses the baseline.

Regarding node embeddings, we report decent results in terms of AUC-ROC and RMSE in all of our experiments: In the facebook dataset we have the best performance in terms of RMSE and  have AUC near the other methods; less than 1\% for all labels, and similar results for top-50\% and top-1. In the dblp-dyn, fb-pages and github\footnote{The dataset contains one label hence AUC-ROC results remain the same.} dataset we outperform the other methods --- with the exception of the AUC-ROC in top-50\% in dblp-dyn where we have a 4\% percent decrease. Moreover, in the fb-pages dataset, GraphWave achieves a very small RMSE however it yields a low AUC-ROC by far. In the pokec network, GraphWave and Random HK fail to run subject to our resources\footnote{Denoted by the dagger ($\dagger$) symbol. Experiments were run in a Google Colab Notebook.}. Moreover, the NNIM model runs two orders of magnitude faster with $k = \lceil \log n \rceil$ neighbors and one order of magnitude faster with $k = \lceil \sqrt n \rceil$ neighbors compared to node2vec and NodeSketch. The PCA step does not affect the runtime considerably needing only 1 sec since it fits only on the highly influential nodes that are $n^{0.7}$, which account for 1.92\% of the network. We achieve an AUC-ROC of 91.84\% and an RMSE of 0.025 where we surpass NodeSketch in terms of RMSE (6 times lower) and are surpassed in terms of AUC-ROC by 0.3\%. Finally, node2vec has a higher AUC-ROC rate (by a small margin) compared to NNIM with $k = \lceil \sqrt n \rceil$ neighbors.

The core-periphery structure of networks is a well-studied problem, but it has gathered limited attention regarding its algorithmic implications. A sublinear-size core can be in general used to speed up algorithms in ML and network science. The idea is to augment fast computation performed in the (sublinear) core, to augmentation in the periphery, which can in general extended in problems regarding community detection, embedding generation shortest paths computation etc, which in general attempts to surpass the bottleneck that the periphery cannot be easily gathered.

In this work, we present a \emph{specific application of the framework} and benefit from the core-periphery structure of OSN and develop inference algorithms for interest prediction using partial information from the core users. We use the core users as initializers of a homophily-inspired evolutionary process between the peripheral users that exchanges opinions according to $k$ nearest neighbors. 
Our algorithm for inference is computationally efficient and has connections to traditional models, such as the HK. We prove that our algorithm converges in finite time and strictly bound the total variation distance from the consensus state. Our method is compared with others and in networks of various sizes and is capable of performing considerably faster with similar and most of the times better results. Another interesting pathway for future work is modifying the current algorithm for identifying the core in an online manner; where we maintain a priority queue where users are ordered according to their in-degree and at each iteration the user with the largest in-degree is dequeued. 

% The algorithm starts by a “famous user” of the network. Access to followers and their in-degrees is easily obtainable via REST API calls in most OSN. 

Closing, the theoretical contributions by themselves do not present any foreseeable socio-economical consequences. From a practical aspect, finding the influencers in a network in terms of their structural properties and them to devise the interests of the rest of the network may have socio-economical consequences. Data provided by these users is usually public (for profit reasons) and thus core users can be used as trend-initializers. However, inferring the interests of peripheral users by using information in a semi-supervised manner may  not be fairly used by external agents.

\begin{table*}[t]
    \centering
    \caption{Experiments with $p = 0.7$, $\gamma = 2$, $D = 10^{-3}$, $\tau = 4$, 10 Trees for LSH and regularization with $\alpha = 1$. Random Seed = 17. The hyperparameters used for the other experiments are as follows: (a) Random HK:  $D = 10^{-3}, \; \varepsilon = \sqrt {d / 2}$; (b) node2vec: $d_{emb} = 128$;
    (c) GraphWave: 200 samples, step size of 0.1, heat coeff. 1.0;
    (d) NodeSketch: $d_{emb} = 32$, 2 iter.}

    \begin{tabular}{lrrrrrrr}
    \toprule
                &   \rb{facebook} &   \rb{dblp-dyn} &   \rb{fb-pages} &   \rb{github} &    \rb{dblp} & \rb{pokec} & \rb{Running Time (s)} \\
    \midrule 
    & \multicolumn{7}{c}{AUC-ROC (all labels)} \\
    \midrule 
  
      Collaborative Filtering (Dynamic) & 76.61      & 66.59      &        70.83     &  59.91   & 72.5  & $\dagger$ & $\dagger$ \\
    
      Collaborative Filtering (Bipartite graph) &    79.51   &    81.13   &            91.74 &  67.85   & 77.61   & 74.75   & $\sim {10^1}$ \\

     node2vec   &    {86.35}   &    87.42   &         84.00      &  67.23   & 69.80 & {96.93} & $\sim 10^3$   \\
     GraphWave  &    86.20   &   86.78   &          70.96   &  45.13   & 69.57 & $\dagger$ & $\dagger$  \\
     NodeSketch &    80.90   &   81.90    &         68.68   &  49.96   & 58.88 & 92.14 & $\sim 10^3$  \\
     Random HK ($k = \lceil \log n \rceil$)  &    85.75   &    86.30    &          71.90    &  50.34   & 68.83  & $\dagger$ & $\dagger$ \\
     NNIM ($k = \lceil \log n \rceil$)      &    84.24   &    88.05   &         {91.86}   &  68.07   & 78.64 & 85.60 & $\sim {10^1}$ \\
     NNIM ($k = \lceil \sqrt n \rceil$)      &    85.82 &     {91.16} &          91.62 &  67.86 &  {81.65} & 91.84 & $\sim 10^2$ \\
     NNIM w/ Reg ($k = \lceil \log n \rceil$) &    84.17   &    87.39   &            91.78 &  {72.31}   & 78.86 & 85.05 & $\sim {10^1}$ \\

    \midrule
    & \multicolumn{7}{c}{AUC-ROC (top 50\% of labels)} \\
    \midrule    
  
     Collaborative Filtering (Dynamic) & 76.61      & 66.59      &        70.83     &  59.91   & 72.5 & $\dagger$  & $\dagger$   \\

     Collaborative Filtering (Bipartite Graph) &    {83.47}   &    86.75   &              {100} &  67.85   & {83.43}   & 77.38 & $\sim {10^1}$   \\

     node2vec   &      54.98 &      {94.92} &            78.69 &    67.23 &  68.53 & {96.94} & $\sim 10^3$ \\
     GraphWave  &      53.97 &      92.91 &            40.11 &    45.13 &  65.70 & $\dagger$  & $\dagger$ \\
     NodeSketch &      55.91 &      92.37 &            46.50  &    49.96 &  58.13 & 92.14 & $\sim 10^3$ \\
     Random HK ($k = \lceil \log n \rceil$)  &      52.82 &      93.10  &            56.14 &    50.34 &  64.49 & $\dagger$ & $\dagger$ \\
     NNIM ($k = \lceil \log n \rceil$)       &      59.08 &      79.32 &            89.00    &    68.27 &  78.69 & 85.80  & $\sim  {10^1}$\\
     NNIM ($k = \lceil \sqrt n \rceil$)     &      58.30  &      90.59 &            88.04 &    67.86 &  80.85 & 91.84 & $\sim 10^2$ \\

     NNIM w/ Reg ($k = \lceil \log n \rceil$) &      59.20  &      81.11 &            88.65 &    {72.31} &  79.10 & 85.05 & $\sim {10^1}$ \\
     
    \midrule
    & \multicolumn{7}{c}{$F_1$-Score (binary predictors)} \\
    \midrule 
    
    Label Propagation & 30.03 & 52.93 & 90.00 & 87.85 & 48.49 & 34.56 & $\sim 10^1$ \\

     \midrule 
    & \multicolumn{6}{c}{RMSE (all labels)} \\
    \midrule    
      
      Collaborative Filtering (Dynamic) &  0.011  &  0.038 &         0.153 &   0.445 &  0.066 & $\dagger$ & $\dagger$ \\
      Collaborative Filtering (Bipartite Graph) &     0.012 &     0.067 &             {4e-17}    &   0.389 &  0.149 &  0.026 & $\sim {10^1}$ \\
     node2vec   &     0.012 &     0.059 &           0.093 &   0.438  &  0.166 & {0.022} & $\sim 10^3$ \\
     GraphWave  &      {0.010}   &     0.052 &    7e-6       &   0.400 &  {0.082} & $\dagger$ & $\dagger$ \\
     NodeSketch &     0.096 &     0.123 &           0.098 &   0.440 &  0.316 & 0.128 & $\sim 10^3$ \\
     Random HK ($k = \lceil \log n \rceil$) &    {0.010} &     0.056 &    {4e-17}     &   0.412 &  0.096 & $\dagger$ & $\dagger$ \\
     NNIM ($k = \lceil \log n \rceil$)      &     0.011 &     0.062  &      {4e-17}      &   0.389 &  0.143 & 0.026 & $\sim  {10^1}$\\
     NNIM ($k = \lceil \sqrt n \rceil$)       &     {0.010} &     {0.050} &          {4e-17}      &   {0.388} &   0.128 &  0.025 & $\sim 10^2$\\
     NNIM w/ Reg ($k = \lceil \log n \rceil $) &     0.012 &     0.066 &             4e-16    &   {0.388} &  0.145 & 0.025 & $\sim  {10^1}$ \\
     Label Propagation &  0.017 &  0.052 &         0.157 &   0.463 &  0.071 &  0.025 & $\sim  {10^1}$ \\

    \midrule
    Coverage (\%) & 88.36 & 97.16 & 72.20 & 68.61 & 66.04 & 51.70 & --- \\ 
    \midrule
    Influencers (Core size) (\%) & 12.47 & 11.83 & 4.94 & 4.23 & 4.12 & 1.92 & --- \\
    \midrule
    Fraction of edges in bipartite graph (\%) & 12.99 & 11.14 & 8.93 & 7.14 & 8.03 & 5.2 & --- \\
    \bottomrule
    \end{tabular}
    \\
    \label{tab:results}
\end{table*}

\bibliographystyle{ACM-Reference-Format}
\bibliography{references.bib}

\appendix

\section{Theoretical Properties of NNIM}

\textbf{Tie Breaking.} We\footnote{All our proofs regarding convergence assume that the model has $d = 1$ dimension (unless otherwise stated), and the coordinate indices are discarded for ease of notation. The results can be extended to $d$ dimensions defining the appropriate structures (convex hull) to showcase cluster isolation phenomena as described below.} define the set $\hat K^{(t)}(u)$ of the $k$ nearest neighbors of $u$. In case of ties, these ties are broken arbitrarily. However, as we prove below, the relative ordering of vertices persists from one round to the next, even if ties are broken arbitrarily
 
\begin{lemma}[Persistence of Relative Ordering] \label{ordering}
    If for two agents $u$ and $v$ at time $t_0$ the relation $\phi_u^{(t_0)} \le \phi_v^{(t_0)}$ holds, then  $\phi_u^{(t)} \le \phi_v^{(t)}$ for all $t \ge t_0$ under arbitrary breaking of ties. 
\end{lemma}

\begin{proof}
    Order the elements of $\hat K^{(t_0)}(u)$ and $\hat K^{(t_0)}(v)$ by their distance to 0. We pick the lefmost element $w \in \hat K^{(t_0)}(u)$, which is related to the leftmost element $z \in K^{(t_0)}(v)$ by the definition of $\hat K^{(t_0)}(u)$ as $\phi_w^{(t_0)} \le \phi_z^{(t_0)}$. We remove the two points and iterate. We finally sum the resulting inequalities and get the result for $t = t_0 + 1$. The case for every $t \ge t_0$ follows inductively. 
\end{proof}

However, an arbitrary tie-breaking mechanism, does not guarantee that our algorithm converges. Hence, we need to devise a systematic ordering under which we resolve ties which we use to prove that our algorithm converges. A natural total ordering $\prec_{i, t}$ of the vertices is to give (globally consistent) ids to the vertices and use the ids as a secondary criterion to break ties in a unique way. The sets $K^{(t)}(u)$ of the $k$ nearest neighbors are defined with respect to the $\prec_{i, t}$ total ordering relation and therefore \emph{ties are eliminated}. We also define  the set $\sigma^{(t)}(u) = \{ v \in U\mid \vec \phi_v^{(t)} = \vec \phi_v^{(t)} \}$. Note that $v_j \prec_{i, t} v_\ell \notimplies v_j \prec_{i, t + 1} v_\ell$. Moreover, we observe that when two agents ``fuse'' together at time $t_0$, they remain fused for all $t \ge t_0$. The set $\sigma^{(t)}(u)$ is monotone, i.e. $t_1 \le t_2 \iff \sigma^{(t_1)}(u) \subseteq \sigma^{(t_2)}(u)$ for all $u \in U$. 

\begin{lemma}[Termination] \label{termination}

The NNIM algorithm converges at time $T \in \mathbb N \cup \{ \infty \} $ if and only if $|\sigma^{(T)}(u)| \ge k$ for every $u \in U$.

\end{lemma}

\begin{proof} 

$(\impliedby)$ This direction is trivial. Let $|\sigma^{(T)}(u)| \ge k$ for all $u \in U$. Then $\sigma^{(T)}(u) \supseteq K^{(T)}(u)$ for all $u \in U$. The result follows by applying the update rule and the definition of $\sigma^{(T)}(u)$. 

\noindent $(\implies)$ Suppose that the NNIM algorithm converges. Equivalently for every $t \ge T$ and for every $w \in U$ we have $\phi_w^{(t)} = \phi_w^{(T)}$. We reside in the case that $t = T + 1$ since the rest follows by induction. Suppose that there exists some $u \in U$ such that $|\sigma^{(T)}(u)| < k$. Then the set $K^{(T)}(u) \setminus \sigma^{(T)}(u)$ is non-empty. So $\phi_u^{(T+1)} = \frac {k - |K^{(T)}(u) \setminus \sigma^{(T)}(u)|} k \phi_u^{(T)} + \frac 1 k  \sum_{w \in K^{(T)}(u) \setminus \sigma^{(T)}(u)} \phi_v^{(T)}
$

From the fact that the system has converged we obtain that
$
 \phi_u^{(T)} = \frac 1 {|K^{(T)}(u) \setminus \sigma^{(T)}(u)|}  \sum_{w \in K^{(T)}(u) \setminus \sigma^{(T)}(u)} \phi_v^{(T)}
$
which yields a \emph{contradiction} since there are no constraints on the values of $\phi_v^{(T)}$ which impose such a relation. Therefore, for every $w \in U$ the set $\sigma^{(T)}(w)$ contains at least $k$ elements. 

\end{proof}

We define the distance of two sets $W, Z \subseteq U$ as the quantity $    \delta_{WZ}^{(t)} = \min_{w \in W, z \in Z} \| \phi_w^{(t)} - \phi_z^{(t)} \| 
$, which satisfies the properties of a metric (non-negativity, identity, symmetry, subadditivity). Moreover we define that two (non-overlapping) intervals \emph{split} if and only if the $k$-nearest neighbor of each of the closest points are less than $\delta_{WZ}^{(t)}$ for some $t \ge 0$. 
\begin{lemma} \label{splitting}
    If two non-overlapping intervals split at $t_0$ then they remain split for all $t \ge t_0$. 
\end{lemma}

\begin{proof}
    Let $W, Z \subseteq U$ be two non-overlapping clusters that have split at $t_0$. Let $\hat w, \hat z$ be the closest points of $W, Z$. Without loss of generality let $\phi_{\hat w}^{(t_0)} < \phi_{\hat z}^{(t_0)}$. Then for all $u \in K^{(t_0)}(w)$ we have that $\phi_{u}^{(t_0)} \le \phi_{\hat w}^{(t_0)}$. By summing up we get $\phi_{\hat w}^{(t_0 + 1)} \le \phi_{\hat w}^{(t_0)}$. Similarly $\phi_{\hat z}^{(t_0)} \le \phi_{\hat z}^{(t_0 + 1)}$. Therefore the minimum distance increases. Hence the sets remain split at $t_0 + 1$. Inductively the sets remain split for all $t \ge t_0$ 
\end{proof}

We define the \emph{splitting time} of $W$ and $Z$ as the minimum $t_0$ that the split occurs. We also define that a subset of (consecutive) agents $W \subseteq U$ of cardinality at least $k$ is said to be \emph{isolated} if and only if there exists some $t_0 \ge 0$ such that it splits from the left set $l(W) = \{ v \in U \setminus W | \phi_v^{(t_0)} < \inf_{w \in W} \phi_w^{(t_0)} \}$ and the right set $r(W) = \{ v \in U \setminus W | \phi_v^{(t_0)} > \sup_{w \in W} \phi_w^{(t_0)} \}$. 

\textbf{Model Convergence.} We now write the system in vector format $\Phi(t + 1) = A(t) \Phi(t)$ where $\Phi(t)$ is the column vector with elements $\phi_u(t)$ and $A(t)$ is the stochastic matrix with entries $A_{uv}(t) = 0$ if $v \not \in K^{(t)}(u)$ and $A_{uv}(t) = 1 / k$ otherwise. We use the following

\begin{theorem}[Theorem 1 of \cite{nedic2012multi}] \label{thm:touri}

Let $ \{A(t) \}$ be a sequence of stochastic matrices such that the following two properties hold: (i) There exists a scalar $\gamma \in (0, 1]$ such that $A_{ii} \ge \gamma$ for all $i \in [n], t \ge 0$; (ii) There exists $\alpha \in (0, 1]$ s.t. $\forall \emptyset \neq S \subset [n]$ and $\bar S = [n] \setminus S$ it holds that $\sum_{i \in S, j \in \bar S} A_{ij}(t) \ge \alpha \sum_{j \in \bar S, i \in S} A_{ji}(t)$ for all $t \ge 0$. Then the dynamics $\Phi(t + 1) = A(t) \Phi(t)$ admit an adjoint sequence with probability vectors which are uniformly bounded away from 0.

\end{theorem}

We prove that our model is globally asymptotically stable (GAS).

\begin{lemma}[GAS] \label{gas}
    The NNIM model is GAS.
\end{lemma}

\begin{proof}

We have that $A_{uu}(t) = \frac 1 k$, and for every element of $A(t)$ we have that $A_{uv}(t) = \frac 1 k \mathbf 1 \{ v \in K^{(t)}(u) \} \ge \frac 1 {nk} \mathbf 1 \{ u \in K^{(t)}(v) \} = \frac 1 n A_{vu}(t)$. Therefore, by summation on a subset $S \subset U$ we get $\sum_{u \in S, v \in \bar S} A_{uv}(t) \ge \frac 1 n \sum_{u \in \bar S, v \in S} A_{vu}(t)$. From Theorem~\ref{thm:touri}, the NNIM dynamics admit an adjoint sequence $\Pi(t) = ( \pi_1(t), \dots, \pi_n(t) )^T$ such that $\Pi^T(t + 1) = \Pi^T(t) A(t)$ with $\pi_u(t) > p$ for all $u$ and some $1 > p > 0$. 

We then define the Lyapunov function
$
    V(t) = \sum_{i = 1}^n \pi_u(t) \| \phi_u(t) - \Pi^T(t) \Phi(t) \|_2^2  
$. Our approach follows the methodology presented in \cite{averaging_lyapunov} and \cite{nedic2012multi}. Note that $V(t) \ge 0$ for all $t \ge 0$. Letting $H(t) = A^T(t) \mathrm{diag}(\pi_u(t + 1)) A(t)$ and doing the matrix operations expressing $V(t)$ as a quadratic form the function $V(t)$ can be written as $V(t) = V(t + 1) + \frac 1 2 \sum_{u, v} H_{uv} (t) (\phi_u^{(t)} - \phi_v^{(t)})^2 $ since $H^T(t) = H(t)$. The elements of $H(t)$ are 
$
    H_{uv}(t) = \frac 1 {k^2} \sum_{w} \pi_w(t + 1) \mathbf 1 \{ u \in K^{(t)} w \} \mathbf 1 \{ v \in K^{(t)} (w) \}
$. Therefore,  combining everything we arrive at 

\begin{equation}
\small
    V(t + 1) = V(t) - \frac {1} {2k^2} \sum_{w} \pi_w (t + 1) \sum_{u, v \in K^{(t)}(w)} (\phi_u^{(t)} - \phi_v^{(t)})^2 \le V(t)
\end{equation}

Hence the function $V(t)$ is decreasing globally in $[0, 1]^d$. Therefore $\lim_{t \to \infty} \Phi(t) = \Phi^*$ (non-infinite).   

\end{proof}

We proceed by proving that convergence indeed occurs in finite time.

\begin{lemma} \label{finite}
    The NNIM Model converges in finite time.
\end{lemma} 

\begin{proof}
Eliminating recurrence via observing that the sum telescopes, we arrive at 

\begin{equation}
    V(t) =  V(0) - \frac {1} {2k^2} \sum_{t = 0}^T \sum_w \pi_w(t + 1) \sum_{u, v \in K^{(t)}(w)} (\phi_u^{(t)} - \phi_v^{(t)})^2
\end{equation}

Since $V(t) \ge 0$ for every $T$, the negative difference term should vanish as $T \to \infty$. More specifically

\begin{equation}
    \lim_{T \to \infty} \frac {1} {2k^2} \sum_{t = 0}^T \sum_w \pi_w(t + 1) \sum_{u, v \in K^{(t)}(w))} (\phi_u^{(t)} - \phi_v^{(t)})^2 = 0  
\end{equation}

Note that $\pi_w(t + 1) > p$ for some $p \in (0, 1)$ by the definition of the adjoint dynamics and $k > 0$, hence we have a sum of squares with positive coefficients vanishing as $T \to \infty$. In order for this to happen, every individual term of the sum must go to 0. Therefore, for every $w \in U$, by the squeeze theorem  $\lim_{T \to \infty} \sum_{u, v \in K^{(T)}(w) \setminus \sigma^{(T)}(w)} (\phi_u^{(T)} - \phi_v^{(T)})^2 = 0  
$. Again by the same argument for all $u, v \in \lim_{T \to \infty} K^{(T)}(w)$ for all $w \in U$ we have that $
    \lim_{T \to \infty} (\phi_u^{(T)} - \phi_v^{(T)}) = 0 
$. By the definition of NNIM, the update process is continuous hence $\lim_{T \to \infty} \phi_u^{(T)} = \lim_{T \to \infty} \phi_v^{(T)}$ as well as by the monotonicity of $V(t)$ we know that there exists some $\phi_w^* \in [0, 1]$ such that $
    \lim_{T \to \infty} \phi_u^{(T)} = \lim_{T \to \infty} \phi_v^{(T)} = \phi_w^*
$. Hence $\lim_{T \to \infty} \phi_u^{(T)} = \phi_w^*$ for all $u \in \lim_{T \to \infty} K^{(T)}(w)$. Therefore for every $\epsilon_w > 0$ there exists some $T_w \ge 0$ such that for all $t \ge T_w$ and for all $u \in K^{(t)}(w)$ we get that   $|\phi_u^{(t)} - \phi_w^* | < \epsilon_w$. Now we prove finite time convergence via choosing the correct values for the $\epsilon$'s.

By Lemma \ref{single} we know that if there exists a unique limiting point then it must be exactly approached in finite time. Suppose that there are $r \ge 2$ distinct limiting points $0 \le \phi^*_1 < \phi^*_2 < \dots < \phi^*_r \le 1$. Now, fix $\epsilon > 0$. We know that for every $w \in U$ and $\epsilon_w = \epsilon$ there exists some finite $T_w \ge 0$ at which $w$ reaches its limiting point within a distance of $\epsilon$. Hence the maximum distance between two elements of $K^{(t)}(w)$ for $t \ge T_w$ is at most $2 \epsilon$, by the triangle inequality, and the same applies for every pair of points that approach this limit. Let $W_1, \dots, W_r \subseteq U$ be the subsets of $U$ that approach their corresponding limits. From Theorem \ref{ordering} these sets must contain consecutive agents. In order for finite convergence to occur we must impose a value of $\epsilon$ which splits the sets from each other. In this way, as we proved in Lemmas \ref{single} and \ref{splitting}, we attain a finite convergence time.

First of all, let $T' = \max_{1 \le m \le r} \max_{w \in W_m} T_w < \infty$ and let $D = \min_{i, j} \delta^{(T')}_{W_i W_j}$. A splitting occurs when the maximum distance between two points reaching the same limit, namely $2 \epsilon$ is less than the minimum distance $D$, hence $2 \epsilon < D$. A good choice for $\epsilon$ is the one which satisfies $2 \epsilon + D < \min_{1 \le i \le r - 1} \{ \phi^*_{i + 1} - \phi^*_i \} $. Therefore, by these two conditions choosing 
$0 < \epsilon < \frac 1 4 \min_{1 \le i \le r - 1} \{ \phi^*_{i + 1} - \phi^*_i \}$ isolates the sets $W_1, \dots, W_r$, hence by Lemma \ref{single} there exist $T_1, \dots, T_r < \infty$ at which each $W_i$ reaches its limit point. Now choose $T = T' + \max_{1 \le i \le r} T_i + 1 < \infty$ and the proof is complete. 

\end{proof}

We provide the helper Lemma below

\begin{lemma} \label{single}
    Suppose that the NNIM approaches (asymptotically) to a unique point $\phi^*$, namely $\lim_{t \to \infty} \Phi(t) = \phi^* \mathbf 1$. Then this point must be reached in finite time.
\end{lemma}

\begin{proof}
    At least one of the leftmost point or the rightmost point must have (in order for the one limit point to exist) a neighbor with different coordinate, to their right or to their left respectively. Since the points have continuous positions with preserved ordering there exists some finite time $0 \le T < \infty$ at which they reach the same point $\phi^*$. 

\end{proof}

\begin{lemma} \label{eig_tv}
    The total variation distance  $d_{TV}(t)$ of the 1D NNIM model decreases as $o(k^{-t/2})$ a.a.s. for $n \to \infty$ and any $k \in \mathbb N$. More specifically, if we fix some small $\delta \in [0, 1]$, and  $n = \Omega \left (\delta^{-1/\tau} \right)$ agents where $\tau = \lceil (\sqrt{k - 1} + 1) / 2 \rceil - 1$ then with probability of at least $1 - \delta$ the total variation distance $d_{TV}(t) = \frac 1 2 \| \Phi(t) - \Phi^* \|_{1}$ decreases as $o(k^{-t/2})$. 
\end{lemma}

\begin{proof}

Let $\lambda_2 (A(t'))$ represent the second largest eigenvalues of the stochastic matrices $A(t')$ and let $\lambda_2^* = \max_{0 \le t' \le t - 1} \lambda_2 (A(t'))$. Then by the Perron-Frobenius Theorem the convergence rate is dominated by the second largest eigenvalue of the ``slowest'' matrix, i.e. $d_{TV}(t) = O \left ( (\lambda_2^*)^t \right )$. We define the matrix sequence $\{ B(t') \}_{0 \le t' \le t-1}$ such that $B(t') = k A(t')$. The matrices $\{ B(t') \}_{0 \le t' \le t - 1}$ represent the adjacency matrices of $k$-regular graphs with self-loops. Hence our problem resides in determining an upper bound on the second largest eigenvalue of a $k$-regular graph $G(t')$. This is a well known problem in Spectral Graph Theory once conjectured by Alon \cite{conjecture} and recently proved by Friedman in reference \cite{friedman}. Alon-Friedman's Theorem states that for any $0 \le t' \le t -  1$ the following concentration bound holds $\Pr [\lambda_2(B(t')) \le 2 \sqrt{k - 1} + \varepsilon ] \ge 1 - O(n^{-\tau})$ for some fixed $\varepsilon > 0$ and $\tau = \lceil (\sqrt{k - 1} + 1) / 2 \rceil - 1$. Therefore
    $\Pr [\lambda_2(A(t')) \in o(k^{-1/2}) ] \ge 1 - O(n^{-\tau})
$ and hence $
    \Pr [|\lambda_2^*| \in o(k^{-1/2}) ] \ge 1 - O(n^{-\tau})
$. For the total variation distance $d_{TV}(t)$ we get that $
    \Pr [d_{TV}(t) \in \Omega(k^{-t/2}) ] \le O(n^{-\tau}) 
$. Finally, choosing $n = \Omega \left (\delta^{- \frac 1  \tau} \right )$ we have that with probability of at least $1 - \delta$ the total variation distance behaves as $o(k^{-t/2})$.

\end{proof}

\textbf{Proof of Theorem \ref{thm:convergence}.} We combine the finite time convergence result of Lemma \ref{finite} and the convergence rate of Lemma \ref{eig_tv}. In the case of the $d$-dimensional model the guarantee translates to a total variation distance of $d \cdot D$.    

\textbf{Concentration Bound for the Hamming Distance (Proof Sketch).} For two $d$-dimensional Bernoulli r.v.s $\vec X, \vec Y$ with independent components and expectations $\vec p, \vec q$ resp. we can apply McDiarmid's ineq.~\cite{doob1940regularity} and the triangle ineq. to conclude w.p. at least $1 - \delta$ the Hamming distance $\| \vec X - \vec Y \|$ varies no more than $d / 2 + O(\sqrt{dlog(1/\delta)})$ from the L2-squared distance $\| \vec p - \vec q \|$.

\end{document}